\pdfoutput=1
\documentclass[submission,copyright,creativecommons]{eptcs}
\providecommand{\event}{Fixed Points in Computer Science 2015 (FICS '15)}

\usepackage[utf8]{inputenc}
\usepackage{amssymb,amsmath}
\usepackage{amsthm}
\usepackage{stmaryrd,colonequals}
\usepackage{mathtools}
\usepackage[f]{esvect}
\usepackage{hhline}
\usepackage{xargs}
\usepackage{xspace}
\usepackage{calc}
\usepackage{tikz}
\usepackage{tikz-cd}
\usepackage{forest}
\usepackage{array}
\usepackage[inline]{enumitem}
\usepackage{bussproofs}
\usepackage{listings}
\usepackage{wrapfig}
\usepackage{ifthen}

\let\amalg=\undefined
\let\coprod=\undefined
\DeclareSymbolFont{cmsymbols}{OMS}{cmsy}{m}{n}
\DeclareSymbolFont{cmlargesymbols}{OMX}{cmex}{m}{n}
\DeclareMathSymbol{\amalg}{\mathbin}{cmsymbols}{"71}
\DeclareMathSymbol{\coprod}{\mathop}{cmlargesymbols}{"60}

\usetikzlibrary{calc,positioning}

\usepackage{local-macros}
\usepackage{thm}
\setenumerate{listparindent=\parindent}

\title{Dependent Inductive and Coinductive Types
  are Fibrational Dialgebras}
\author{Henning Basold
\institute{Radboud University, iCIS, Intelligent Systems}
\institute{CWI, Amsterdam, The Netherlands}
\email{h.basold@cs.ru.nl}
}
\def\titlerunning{Dependent Dialgebras}
\def\authorrunning{H. Basold}

\begin{document}
\maketitle
% \let\thefootnote\relax
% \footnotetext{Base revision~\gitAbbrevHash, \gitReferences~from~\gitAuthorDate}

\begin{abstract}
  In this paper, I establish the categorical structure necessary to interpret
  dependent inductive and coinductive types.
  It is well-known that dependent type theories à la Martin-Löf can be
  interpreted using fibrations.
  Modern theorem provers, however, are based on more sophisticated type systems
  that allow the definition of powerful inductive dependent types
  (known as inductive families) and, somewhat limited, coinductive dependent
  types.
  I define a class of functors on fibrations and show how data
  type definitions correspond to initial and final dialgebras for these
  functors.
  This description is also a proposal of how coinductive types
  should be treated in type theories, as they appear here simply as dual of
  inductive types.
  Finally, I show how dependent data types correspond to algebras and
  coalgebras, and give the correspondence to dependent polynomial functors.
\end{abstract}

\section{Introduction}
\label{sec:intro}

It is a well-established fact that the semantics of inductive data types without
term dependencies can be given by initial algebras, whereas the semantics of
coinductive types can be given by final coalgebras.
However, for types that depend on terms, the situation is not as clear-cut.

Partial answers for inductive types can be found
in~\cite{Altenkirch-IndexedCont,Chapman:GentleArtLevitation,%
  Dagand-McBride:Ornaments,Gambino-WTypesPolyFunc,HamanaFiore-GADT,%
  Loh-Magalhaes:IndexedFunctors,Moerdijk-WellFoundedTrees},
where semantics have been given for inductive types through polynomial functors
in the category of set families or in locally Cartesian closed categories.
Similarly, semantics for non-dependent coinductive types have been given
in~\cite{AbbottThesis,AbbottContainer,vdBerg-Non-wellfoundedTrees} by using
polynomial functors on locally Cartesian closed categories.
Finally, an interpretation for Martin-Löf type theory (without recursive
type definitions) has been given in~\cite{Seely-LCCC-Types} and corrected
in~\cite{Hofmann-LCCC-Types}.

So far, we are, however, lacking a full picture of dependent coinductive types
that arise as duals of dependent inductive types.
To actually get such a picture, I extend in the present work
Hagino's idea~\cite{Hagino-Dialg}, of using dialgebras to describe data types,
to dependent types.
This emphasises the actual structure behind (co)inductive types as their
are used in systems like Agda.\footnote{It should be noted that, for example,
  Coq treats coinductive types differently. In fact, the route taken in Agda
  with copatterns and in this work is much better behaved.}
Moreover, dialgebras allow for a direct interpretation of types in this
categorical setup, without going through translations into, for example,
polynomial functors.

Having defined the structures we need to interpret dependent data types,
it is natural to ask whether this structure is actually sensible.
The idea, pursued here, is that we want to obtain initial and final dialgebras
from initial algebras and final coalgebras for polynomial functors.
This is achieved by showing that the dialgebras in this work
correspond to algebras and coalgebras, and that their fixed points can be
constructed from fixed points of polynomial functors (in the sense of
\cite{GambinoKockPolyFunctors}).
% The reduction to polynomial functors does not work for all types for the time
% being, but rather only  for those that have an arbitrary nesting
% of non-dependent coinductive and dependent inductive types, and a dependent
% coinductive type on the top-level.

To summarise, this paper makes the following contributions.
First, we get a precise description of the categorical structure necessary to
interpret inductive and coinductive data types, which can be seen as
categorical semantics for an extension of the inductive and (copattern-based)
coinductive types of Agda.
The second contribution is a reduction to fixed points of polynomial functors.

What has been left out, because of space constraints, is an analysis of
the structures needed to obtain induction and coinduction principles.
Moreover, to be able to get a sound interpretation, with respect to type
equality of dependent types, we need to require a Beck-Chevalley condition.
This condition can be formulated for general (co)inductive types,
but is also not given here.

\begin{description}
\item[Related work] As already mentioned, there is an enormous body of work
  on obtaining semantics for (dependent) inductive, and to some extent,
  coinductive types, see~%
  \cite{Altenkirch-IndexedCont,Gambino-WTypesPolyFunc,HamanaFiore-GADT,%
    Moerdijk-WellFoundedTrees}.
  In the present work, we will mostly draw from \cite{AbbottContainer}
  and \cite{GambinoKockPolyFunctors}.
  Categorical semantics for basic Martin-Löf type theory have been developed,
  for example, in~\cite{Hofmann-LCCC-Types}.
  An interpretation, closer to the present work, is given in terms of fibrations
  by Jacobs~\cite{Jacobs1999-CLTT}.
  In the first part of the paper, we develop everything on rather arbitrary
  fibrations, which makes the involved structure more apparent.
  Only in the second part, where we reduce data types to polynomial functors,
  we will work with slice categories, since most of the work on polynomial
  functors in that setting~\cite{AbbottContainer,GambinoKockPolyFunctors}.
  Last, but not least, the starting idea of this paper is of course inspired
  by the dialgebras of Hagino~\cite{Hagino-Dialg}.
  These have also been applied to give semantics to
  induction-induction~\cite{Altenkirch-IndInd} schemes.
\item[Outline]
  The rest of the paper is structured as follows.
  In \secRef{fib-dialg}, we analyse a typical example of a dependent inductive
  type, namely vectors, that is, lists indexed by their length.
  We develop from this example a description of inductive and coinductive
  dependent data types in terms of dialgebras in fibrations.
  This leads to the requirements on a fibration, given in \secRef{dt-complete},
  that allow the interpretation of data types.
  In the same section, we show how dependent and fibre-wise (co)products arise
  canonically in such a structure, and we give an example of a coinductive type
  (partial streams) that can only be treated in Agda through a cumbersome
  encoding.
  The reduction of dependent data types to polynomial functors is carried out
  in \secRef{construct-dt}, and finish with concluding remarks in
  \secRef{concl}.
\item[Acknowledgement]
  I would like to thank the anonymous reviewers, who gave very valuable feedback
  and pointed me to some more literature.
\end{description}

% The paper is structured as follows.
% In \secRef{fib-dialg}, we analyse a typical example of a dependent inductive
% type: vectors, that is, lists indexed by their length.
% We develop from this example a description of inductive and coinductive
% dependent data types in terms of dialgebras in fibrations, and show how these
% relate to algebras and coalgebras.
% Next, we establish in \secRef{dt-complete} the requirements on a fibration that
% allow the interpretation of data types, and show how dependent and
% fibre-wise (co)products arise.
% \secRef{b-c-cond} is concerned with the definition and consequences of a
% Beck-Chevalley condition for data types.
% Having given a precise account for the definition of data types, we go on and
% study induction and coinduction principles in \secRef{co-induction}.
% We give simple versions of these principles already derivable without further
% assumption, and then analyse the precise extra requirements we need to make
% induction work.
% The last part largely follows
% \cite{HermidaJacobs-IndFib} and \cite{Ghani-IndexedInduct}.

% \paragraph*{Acknowledgements}
% I would like to thank Daniela Petrișan for providing me insights in fibred
% adjunctions, and Herman Geuvers for the fruitful discussions about the
% present work and logic in general.

%%% Local Variables:
%%% TeX-master: "../FibDialg-FICS"
%%% ispell-local-dictionary: "british"
%%% mode: latex
%%% End:

\section{Fibrations and Dependent Data Types}
\label{sec:fib-dialg}

In this section we introduce \emph{dependent data types} as initial and
final dialgebras of certain functors on fibres of fibrations.
We go through this setup step by step.

Let us start with dialgebras and their homomorphisms.
\begin{definition}
  \label{def:dialgebra}
  Let $\Cat{C}$ and $\Cat{D}$ be categories and $F, G : \Cat{C} \to \Cat{D}$
  functors.
  An $(F, G)$-\emph{dialgebra} is a morphism $c : F A \to G A$ in $\Cat{D}$,
  where $A$ is an object in $\Cat{C}$.
  Given dialgebras $c : F A \to G A$ and $d : F B \to G B$, a morphism
  $h : A \to B$ is said to be a (dialgebra) \emph{homomorphism} from $c$ to $d$,
  if $G h \, \circ \, c = d \, \circ \, F h$.
  This allows us to form a category $\DialgC{F,G}$, in which objects
  are pairs $(A, c)$ with $A \in \Cat{C}$ and $c : F A \to G A$, and morphisms
  are dialgebra homomorphisms.
\end{definition}

The following example shows that dialgebras arise naturally from data types.
\begin{example}
  \label{ex:vectors-intro}
  Let $A$ be a set, we denote by $A^n$ the $n$-fold product of $A$, that is,
  lists of length $n$.
  Vectors over $A$ are given by the set family
  $\VecT A = \{A^n\}_{n \in \N}$, which is an object in the category $\SetC^\N$
  of families indexed by $\N$.
  In general, this category is given for a set $I$ by
  \begin{equation*}
    \SetC^I =
    \begin{cases}
      \text{objects } & X = \{X_i\}_{i \in I} \\
      \text{morphisms }  & f = \{f_i : X_i \to Y_i\}_{i \in I}
    \end{cases}.
  \end{equation*}

  Vectors come with two constructors: $\nil : \T \to A^0$ for the empty vector
  and prefixing $\cons_n : A \times A^n \to A^{n+1}$ of vectors with elements
  of $A$.
  We note that $\nil : \{\T\} \to \{A^0\}$ is a morphism in the category
  $\SetC^\T$ of families indexed by the one-element set $\T$, whereas
  $\cons = \{\cons_n\} : \{A \times A^n\}_{n \in \N} \to \{A^{n+1}\}_{n \in \N}$
  is a morphism in $\SetC^\N$.

  Let $F, G : \SetC^\N \to \SetC^{\T} \times \SetC^{\N}$ be the functors into
  the product of $\SetC^\T$ and $\SetC^\N$ with
  \begin{align*}
    F(X) = \lparen \{\T\}, \{A \times X_n\}_{n \in \N})
    \qquad G(X) = (\{X_0\}, \{X_{n+1}\}_{n \in \N}).
  \end{align*}
  Using these, we find that $(\nil, \cons) : F(\VecT A) \to G(\VecT A)$
  is an $(F, G)$-dialgebra, in fact, it is the \emph{initial} $(F,G)$-dialgebra.
\end{example}

\begin{definition}
  An $(F,G)$-dialgebra $c : F A \to G A$ is called \emph{initial}, if for
  every $(F,G)$-dialgebra $d : F B \to G B$ there is a unique homomorphism $h$
  from $c$ to $d$, the \emph{inductive extension} of $d$.
  Dually, $(A, c)$ is \emph{final}, provided there is a unique homomorphism
  $h$ from any other dialgebra $(B,d)$ into $c$.
  Here, $h$ is the \emph{coinductive extension} of $d$.
\end{definition}

Having found the algebraic structure underlying vectors, we continue by
exploring how we can handle the change of indices in the constructors.
It turns out that this is most conveniently done by using fibrations.
\begin{definition}
  \label{def:fibration}
  Let $P : \TCat \to \BCat$ be a functor, where the $\TCat$ is called the
  \emph{total} category and $\BCat$ the \emph{base} category.
  A morphism $f : A \to B$ in $\TCat$ is said to be \emph{cartesian over}
  $u : I \to J$, provided that
  \begin{enumerate*}[label=\roman*)]
  \item $P f = u$, and
  \item for all $g : C \to B$ in $\TCat$ and $v : PC \to I$ with
    $Pg = u \circ v$ there is a unique $h : C \to A$ such that $f \circ h = g$.
  \end{enumerate*}
  For $P$ to be a \emph{fibration}, we require that
  for every $B \in \TCat$ and $u : I \to PB$ in $\Cat{B}$, there is
  a cartesian morphism $f : A \to B$ over $u$.
  Finally, a fibration is \emph{cloven}, if it comes with a unique choice
  for $A$ and $f$, in which case we denote $A$ by $\reidx{u}B$ and
  $f$ by $\cartL{u} B$, as displayed in the diagram on the right.
\end{definition}

\begin{wrapfigure}[6]{r}{.32\textwidth}
\vspace{-1.8\baselineskip}
%\hspace*{-20pt}
\begin{equation*}
  \begin{tikzcd}[row sep=0.1cm, column sep=0.25cm]
    C \arrow[bend left=15]{drrr}{g}
      \arrow[dashed,shorten >= -5pt]{dr}[swap]{!h}
      %\ar[dd, dotted, no head]
      & & \\
    & \reidx{u} B
      \arrow{rr}[swap]{\cartL{u} B}
      %\ar[dd, dotted, no head]
    & & B
      %\ar[dd, dotted, no head]
      & \TCat \arrow{dddd}{P}
    \\ \\ \\
    PC
      \arrow[bend left=15]{drrr}{P g}
      \arrow{dr}[swap]{v} & & \\
    & I
      \arrow{rr}[swap]{u} & & PB
      & \BCat
  \end{tikzcd}
\end{equation*}
\end{wrapfigure}

At first sight, this definition is arguably intimidating to someone who has
never been exposed to fibrations.
The idea is that the base category $\BCat$ contains as objects the indices
of objects in $\TCat$, and as morphisms substitutions.
The result of carrying out a substitution on indices, is captured by
the Cartesian lifting property.
Let us illustrate this on set families.
We define $\Fam{\SetC}$ to be the category
\begin{equation*}
  \Fam{\SetC} =
  \begin{cases}
    \text{objects } & (I, X : I \to \SetC), \, I \text{ a set} \\
    \text{morphisms } & (u, f) : (I, X) \to (J, Y) \text{ with }
    u : I \to J \text{ and }
    \{f_i : X_i \to Y_{u(i)}\}_{i \in I}
  \end{cases}
\end{equation*}
in which composition is defined by
\begin{equation*}
  (v, g) \circ (u, f) = \left( v \circ u,
    \{X_i \xrightarrow{f_i} Y_{u(i)} \xrightarrow{g_{u(i)}} Z_{v(u(i))}\}_{i \in I}
  \right).
\end{equation*}
A concrete object is the pair $(\N, \VecT A)$, where $\VecT A$
is the family of vectors from \iExRef{vectors-intro}.

We define a cloven fibration on set families.
Let $P : \Fam{\SetC} \to \SetC$ be the projection on the first
component, that is, $P(I, X) = I$ and $P(u, f) = u$.
For a family $(J, Y)$ and a function $u : I \to J$, we define
$\reidx{u} Y = \{Y_{u(i)}\}_{i \in I}$ and
$\cartL{u} Y = \left(u, \{\id : Y_{u(i)} \to Y_{u(i)}\}_{i \in I} \right)$.
Then, for each $(w, g) : (K, Z) \to (J,Y)$ and $v : K \to I$ with
$w = u \circ v$, we can define the morphism
$(K, Z) \to (I, \reidx{u}Y)$ to be $(v, h)$ with
$h_k : Z_k \to Y_{u(v(k))}$ and $h_k = g_k$, since $u(v(k)) = w(k)$.

An important concept is the \emph{fibre above} an object $I \in \BCat$, given by
the category
\begin{equation*}
  \Cat{P}_I =
  \begin{cases}
    \text{objects } & A \in \Cat{E} \text{ with } P(A) = I \\
    \text{morphisms } & f : A \to B \text{ with } P(f) = \id_I
  \end{cases}.
\end{equation*}
In a cloven fibration, we can use the Cartesian lifting to define for each
$u : I \to J$ in $\Cat{B}$ a functor $\reidx{u} : \Cat{P}_J \to \Cat{P}_I$,
together with natural isomorphisms $\Id_{\Cat{P}_I} \cong \reidx{\id_I}$ and
$\reidx{u} \circ \reidx{v} \cong \reidx{(v \circ u)}$,
see~\cite[Sec. 1.4]{Jacobs1999-CLTT}.
The functor $\reidx{u}$ is called \emph{reindexing} along $u$.

\begin{assumption}
  We assume all fibrations to be cloven in this work.
\end{assumption}

We are now in the position to take a more abstract look at our initial example.
\begin{example}
  First, we note that the fibre of $\Fam{\SetC}$ above $I$ is isomorphic to
  $\SetC^I$.
  Let then $z : \T \to \N$ and $s : \N \to \N$ be $z(\ast) = 0$ and
  $s(n) = n + 1$, giving us reindexing functors
  $\reidx{z} : \SetC^\N \to \SetC^\T$ and $\reidx{s} : \SetC^\N \to \SetC^\N$.
  By their definition, $\reidx{z}(X) = \{X_0\}$ and
  $\reidx{s}(X) = \{X_{n+1}\}_{n \in \N}$, hence the functor $G$, we used to
  describe vectors as dialgebra, is
  $G = \prodArr{\reidx{z}, \reidx{s}}$.
  In \iSecRef{dt-complete}, we address the structure of $F$.
\end{example}

We generalise this situation to account for arbitrary data types.
\begin{definition}
  Let $P : \Cat{E} \to \Cat{B}$ be a fibration.
  A \emph{(dependent) data type signature}, parameterised by a category $\Cat{C}$,
  is a pair $(F, u)$ consisting of
  \begin{itemize}
  \item a functor $F : \Cat{C} \times \Cat{P}_I \to \Cat{D}$ with
    $\Cat{D} = \prod_{k = 1}^n \Cat{P}_{J_k}$ for some $n \in \N$ and $J_k,I \in \BCat$, and
  \item a family $u$ of $n$ morphisms in $\BCat$ with
    $u_k : J_k \to I$ for $k = 1, \dotsc, n$.
  \end{itemize}
\end{definition}
A family $u$ as above induces a functor
$\prodArr{\reidx{u_1}, \dotsc, \reidx{u_n}} : \Cat{P}_I \to \Cat{D}$,
which we will often denote by $G_u$.
This will enable us to define data types for such signatures, but
let us first look at an example for the case $\Cat{C} = \T$, that is, if
$F : \Cat{P}_I \to \Cat{D}$ is not parameterised.
\begin{example}
  \label{ex:dep-sum-prod-pw}
  A fibration $P : \TCat \to \BCat$ is said to have dependent coproducts and
  products, if for each $f : I \to J$ in $\Cat{B}$ there are functors
  $\coprod_f$ and $\prod_f$ from $\Cat{P}_{I}$ to $\Cat{P}_{J}$ that are
  respectively left and right adjoint to $\reidx{f}$.
  For each $X \in \Cat{P}_I$, we can define a signature, such that
  $\coprod_f(X)$ and $\prod_f(X)$ arise as data types for these signatures,
  as follows.
  Define the constant functor
  \begin{align*}
    K_X : \Cat{P}_{J} \to \Cat{P}_I
    \qquad K_X(Y) = X
    \qquad K_X(g) = \id_X.
  \end{align*}
  Then $(K_X, f)$ is the signature for coproducts and products.
  For example, the unit $\eta$ of the
  adjunction $\coprod_f \dashv \reidx{f}$ will be the initial
  $(K_X,\reidx{f})$-dialgebra
  $\eta_X : K_X(\coprod_f(X)) \to \reidx{f}(\coprod_f(X))$, using that
  $K_X(\coprod_f(X)) = X$.
  We come back to this in \iExRef{dep-sum-prod}.
 \qed
\end{example}

To define data types in general, we allow them to have additional parameters,
that is, we allow signatures $(F, u)$, where
$F : \Cat{C} \times \Cat{P}_I \to \Cat{D}$ and $\Cat{C}$ is a non-trivial
category.
Let us first fix some notation.
We put $F(V, -)(X) = F(V, X)$ for $V \in \Cat{C}$, which is a functor
$\Cat{P}_I \to \Cat{D}$.
Assume that the initial
$(F(V, -), G_u)$-dialgebra $\alpha_V : F(V, \Phi_V) \to G_u(\Phi_V)$ and final
$(G_u, F(V, -))$-dialgebra $\xi_V : G_u(\Omega_V) \to F(V,\Omega_V)$ exist.
Then we can define functors
$\mu (\funCur{F}, \funCur{G_u}) :  \Cat{C} \to \Cat{P}_I$
and $\nu (\funCur{G_u}, \funCur{F}) : \Cat{C} \to \Cat{P}_I$,
analogous to~\cite{Jiho2010}, by
\begin{align*}
  & \mu (\funCur{F}, \funCur{G_u})(V) = \Phi_V
  & & \mu (\funCur{F}, \funCur{G_u})(f : V \to W) =
    \indExtW{\alpha_W \circ F(f, \id_{\Phi_W})} \\
  & \nu (\funCur{G_u}, \funCur{F})(V) = \Omega_V
  & & \nu (\funCur{G_u}, \funCur{F})(f : V \to W) =
    \coindExtW{F(f, \id_{\Omega_V}) \circ \xi_V},
\end{align*}
where the bar and tilde superscripts denote the inductive and coinductive
extensions, that is, the unique homomorphism given by initiality
and finality, respectively.
The reason for the notation $\mu (\funCur{F}, \funCur{G_u})$ and
$\nu (\funCur{G_u}, \funCur{F})$ is that these are initial and final dialgebras
for the functors
\begin{align*}
  \funCur{F}, \funCur{G_u} :
    \FunCat{\Cat{C}}{\Cat{P}_I} \to \FunCat{\Cat{C}}{\Cat{D}}
  \qquad \funCur{F}(H) = F \circ \prodArr{\Id_{\Cat{C}}, H}
  \qquad \funCur{G_u}(H) = G_u \circ H
\end{align*}
on functor categories.
That the families $\alpha_V$ and $\xi_V$ are natural in $V$ follows directly
from the definition of the functorial action as (co)inductive extensions.
Hence, they give rise to dialgebras
$\nat{\alpha}{\funCur{F}(\mu (\funCur{F}, \funCur{G_u}))}
  {\funCur{G_u}(\mu (\funCur{F}, \funCur{G_u}))}$
and
$\nat{\xi}{\funCur{G_u}(\nu (\funCur{G_u}, \funCur{F}))}
  {\funCur{F}(\nu (\funCur{G_u}, \funCur{F}))}$.

\begin{definition}
  Let $(F,u)$ be a data type signature.
  An \emph{inductive data type} (IDT) for $(F, u)$ is an initial
  $(\funCur{F}, \funCur{G_u})$-dialgebra with carrier
  $\mu  (\funCur{F}, \funCur{G_u})$.
  Dually, a coinductive data type (CDT) for $(F,u)$ is a final
  $(\funCur{G_u}, \funCur{F})$-dialgebra, note the order, with the
  carrier being denoted by $\nu (\funCur{G_u}, \funCur{F})$.
  If $\Cat{C} = \T$, we drop the hats from the notation.
\end{definition}

\begin{example}
  \label{ex:dep-sum-prod}
  We turn the definition of the product and coproduct from
  \iExRef{dep-sum-prod-pw} into actual functors.
  The observation we use is that the projection functor
  $\pi_1 : \Cat{P}_I \times \Cat{P}_J \to \Cat{P}_I$ gives us a
  ``parameterised'' constant functor: $K_A^J = \pi_1(A, -)$.
  If we are given $f : I \to J$ in $\Cat{B}$, then we use the signature
  $(\pi_1, f)$, and define
  $\coprod_f =  \mu (\funCur{\pi_1}, \funCur{\reidx{f}})$ and
  $\prod_f = \nu (\funCur{\reidx{f}}, \funCur{\pi_1})$.
  We check the details of this definition in \iThmRef{adjunctions-from-dt}.
\end{example}

%%% Local Variables:
%%% TeX-master: "../FibDialg"
%%% ispell-local-dictionary: "british"
%%% mode: latex
%%% End:
\section{Data Type Completeness}
\label{sec:dt-complete}

We now define a class of signatures and functors that should be seen as
categorical language for, what is usually called, strictly positive
types~\cite{Altenkirch-IndexedCont}, positive generalised abstract data
types~\cite{HamanaFiore-GADT} or
descriptions~\cite{Chapman:GentleArtLevitation,Dagand-McBride:Ornaments}.
Note, however, that none of these treat coinductive types.
A \emph{non-dependent} version of strictly positive types that include
coinductive types are given in~\cite{AbbottContainer}.

Let us first introduce some notation.
Given categories $\Cat{C}_1$ and $\Cat{C}_2$ and an object $A \in \Cat{C}_1$, we
denote by $K_A^{\Cat{C}_1} : \Cat{C}_1 \to \Cat{C}_2$ the functor mapping
constantly to $A$.
The projections on product categories are denoted, as usual, by
$\pi_k : \Cat{C}_1 \times \Cat{C}_2 \to \Cat{C}_k$.
Using these notations, we can define what we understand to be a data type
by mutual induction.
\begin{definition}
  \label{def:dt-complete}
  A fibration $P : \TCat \to \BCat$ is \emph{data type complete}, if all IDTs
  and CDTs for \emph{strictly positive signatures} $(F, u) \in \SPSig$ exist,
  where $\SPSig$ is given by the following rule.

  \begin{equation*}
    \AxiomC{$\Cat{D} = \prod_{i = 1}^n \Cat{P}_{J_i}$}
    \AxiomC{$F \in \SPDT[\Cat{C} \times \Cat{P}_I \to \Cat{D}]$}
    \AxiomC{$u = (u_1 : J_1 \to I, \dotsc, u_n : J_n \to I)$}
    \TrinaryInfC{$(F, u) \in \SPSig[\Cat{C} \times \Cat{P}_I \to \Cat{D}]$}
    \DisplayProof
  \end{equation*}
  The functors in $\SPDT$ are given by the following rules,
  assuming that $P$ is data type complete.

  \begin{gather*}
    \AxiomC{$A \in \Cat{P}_J$}
    \UnaryInfC{$K_A^{\Cat{P}_I} \in \SPDT[\Cat{P}_I \to \Cat{P}_J]$}
    \DisplayProof
    \quad
    \AxiomC{$\Cat{C} = \prod_{i=1}^n \Cat{P}_{I_i}$}
    \UnaryInfC{$\pi_k \in \SPDT[\Cat{C} \to \Cat{P}_{I_k}]$}
    \DisplayProof
    \quad
    \AxiomC{$f : J \to I \text{ in } \Cat{B}$}
    \UnaryInfC{$\reidx{f} \in \SPDT[\Cat{P}_I \to \Cat{P}_J]$}
    \DisplayProof
    \quad
    \AxiomC{$F_1 \in \SPDT[\Cat{P}_I \to \Cat{P}_K]$}
    \AxiomC{$F_2 \in \SPDT[\Cat{P}_K \to \Cat{P}_J]$}
    \BinaryInfC{$F_2 \circ F_1 \in \SPDT[\Cat{P}_I \to \Cat{P}_J]$}
    \DisplayProof
    \\[7pt]
    \AxiomC{$F_i \in \SPDT[\Cat{P}_I \to \Cat{P}_{J_i}]$}
    \AxiomC{$i = 1,2$}
    \BinaryInfC{$\prodArr{F_1, F_2}
      \in \SPDT[\Cat{P}_I \to \Cat{P}_{J_1} \times \Cat{P}_{J_2}]$}
    \DisplayProof
    \quad
    \AxiomC{$(F, u) \in \SPSig[\Cat{C} \times \Cat{P}_I \to \Cat{D}]$}
    \UnaryInfC{$\mu (\funCur{F}, \funCur{G_u}) \in \SPDT[\Cat{C} \to \Cat{P}_I]$}
    \DisplayProof
    \quad
    \AxiomC{$(F, u) \in \SPSig[\Cat{C} \times \Cat{P}_I \to \Cat{D}]$}
    \UnaryInfC{$\nu (\funCur{G_u}, \funCur{F}) \in \SPDT[\Cat{C} \to \Cat{P}_I]$}
    \DisplayProof
  \end{gather*}
  This mutual induction is well-defined, as it can be stratified
  in the nesting of fixed points.
\end{definition}

As a first sanity check, we show that a data type complete fibration has, both,
fibrewise and dependent (co)products.
These are instances of the following, more general, result.
\begin{theorem}
  \label{thm:adjunctions-from-dt}
  Suppose $P : \TCat \to \BCat$ is a data type complete fibration.
  Let $\Cat{C} = \prod_{i=1}^m \Cat{P}_{K_i}$ and
  $\pi_1 : \Cat{C} \times \Cat{P}_I \to \Cat{C}$ be the first projection.
  If $G_u : \Cat{P}_I \to \Cat{C}$ is such that $(\pi_1, u)$ is a
  signature, then we have the following adjoint situation:
  \begin{equation*}
    \mu (\funCur{\pi_1}, \funCur{G_u}) \dashv G_u
    \dashv \nu (\funCur{G_u}, \funCur{\pi_1}).
  \end{equation*}
\end{theorem}
\begin{proof}
  We only show how the adjoint transposes are obtained in the case of inductive
  types.
  Concretely, for a tuple $V \in \Cat{C}$ and an object $A \in \Cat{P}_I$,
  we need to prove the correspondence
  \begin{equation*}
    \begin{adjunction}
%      \SwapAboveDisplaySkip
      f : & \mu (\funCur{\pi_1}, \funCur{G_u})(V) & A & \text{in } \Cat{P}_I \\ \hhline{====}
      g : & V & G_u A & \text{in } \Cat{C}
    \end{adjunction}
  \end{equation*}
  Let us use the notation $H = \mu (\funCur{\pi_1}, \funCur{G_u})$, then
  the choice of $\pi_1$ implies that the initial
  $(\funCur{\pi_1}, \funCur{G_u})$-dialgebra is of type
  $\nat{\alpha}{\Id_{\Cat{C}}}{G_u \circ H}$, since
  $\funCur{\pi_1}(H) = \pi_1 \circ \prodArr{\Id_\Cat{C}, H} = \Id_{\Cat{C}}$
  and $\funCur{G_u}(H) = G_u \circ H$.
  This allows us to use as transpose of $f$ the morphism
  $V \xrightarrow{\alpha_V} G_u(H(V)) \xrightarrow{G_uf} G_uA$.
  As transpose of $g$, we use the inductive extension of
  $\funCur{\pi_1}(K_A^{\Cat{C}})(V) = V
  \xrightarrow{g} G_uA = \funCur{G_u}(K_A^\Cat{C})(V)$.
  The proof that this correspondence is natural and bijective follows
  straightforwardly from initiality.
  For coinductive types, the result is given by duality.
\end{proof}

This gives fibrewise coproducts by
$+_I = \mu (\funCur{\pi_1}, \funCur{G_u})$ and products by
$\times_I = \nu (\funCur{G_u}, \funCur{\pi_1})$, using
$u = (\id_I, \id_I)$.
Dependent (co)products along $f : I \to J$ use $u = f$, see
\iExRef{dep-sum-prod}.

There are many more examples of data types that exist in a data type complete
fibration.
We describe three fundamental ones.
\begin{example}
  \label{ex:dt-complete}
  \begin{enumerate}
  \item The first example are initial and final objects inside
    the fibres $\Cat{P}_I$.
    Since an initial object is characterised by having a unique morphism
    \emph{to} every other object, we define it as an initial
    dialgebra, namely $\inObj_I = \mu(\Id, \reidx{\id_I})$.
    Then there is, for each $A \in \Cat{P}_I$, a unique morphism
    $!^A : \inObj_I \to A$ given as inductive extension of $\id_A$.
    Dually, we define the terminal object $\T_I$ in $\Cat{P}_I$ to be
    $\nu(\reidx{\id_I}, \Id)$ and for each $A$ the corresponding unique
    morphism $!_A : A \to \T_I$ as the coinductive extension of $\id_A$.

    Note that this also follows from \iThmRef{adjunctions-from-dt}, if we
    require that (co)inductive data types also exist if
    $\Cat{C} = \T$ (the empty product) and $u = \{\}$ (empty family of
    morphisms).
    This allows us to define the initial and final object as functors
    $\T \to \Cat{P}_I$.

    % Note that these definitions are different from the ones usually chosen
    % in systems like Coq.
    % There, the initial and final object are usually defined as inductive types
    % without constructors or one trivial constructor, respectively.
    % Such a definition of initial objects could be adapted, if we require that
    % inductive data types also exist for empty signatures, that is, if
    % $F, G : \Cat{P}_I \to \Cat{1}$ for the final category $\Cat{1}$.
    % The alternative definition of final objects is more problematic, as
    % the single constructor of $\T_I$ would need a domain.
    % Since this domain must be canonical, we would have to use $\T_I$ itself,
    % which leads to a circular definition, hence an empty inductive type
    % (note that $\inObj_I$ is defined this way).

  % \item
  %   Final coalgebras of endofunctors $\Endo{F}{\Cat{P}_I}$ are just final
  %   $(\Id, F)$-dialgebras.
  %   Since $\reidx{\id_I} \cong \Id_{\Cat{P}_I}$, they are given as coinductive
  %   data type for the signature $(F, \reidx{\id_I})$.
  %   Dually, initial algebras of $F$ are inductive data types for the same
  %   signature.
  % \item (Fibrewise) (co)products like Hagino's description.
  %   \todo{Write}
  \item
    There are several definable notions of equality, provided that $\Cat{B}$
    has binary products.
    A generic one is propositional equality
    $\Eq \; : \, \Cat{P}_I \to \Cat{P}_{I \times I}$, the left adjoint to the
    contraction functor $\reidx{\delta} : \Cat{P}_{I \times I} \to \Cat{P}_I$,
    which is induced by the diagonal $\delta : I \to I \times I$.
    Thus it is given by the dependent coproduct $\Eq \, = \, \coprod_{\delta}$
    and the constructor $\refl_X : X \to \reidx{\delta}(\Eq X)$.
  \item
    Assume that there is an object $\Str{A}$ in $\BCat$ of streams over $A$,
    together with projections to head and tail.
    Then we can define bisimilarity between streams as CDT for the signature
    \begin{align*}
      & F,G_u : \Cat{P}_{(\Str{A})^2} \to
      \Cat{P}_{(\Str{A})^2} \times \Cat{P}_{(\Str{A})^2} \\
      & F = \prodArr*{\reidx{(\head \times \head)} \circ K_{\Eq(A)},
        \reidx{(\tail \times \tail)}}
      \qquad \text{and} \qquad
      u = (\id_{\Str{A} \times \Str{A}},
        \id_{\Str{A} \times \Str{A}}).
    \end{align*}
    Note that there is a category $\Rel{\TCat}$ of binary relations in $\Cat{E}$
    by forming the pullback of $P$ along $\Delta : \BCat \to \BCat$ with
    $\Delta(I) = I \times I$, see~\cite{HermidaJacobs-IndFib}.
  %   \begin{wrapfigure}[5]{r}{.35\textwidth}
  %     \vspace{-1.8\baselineskip}
  %   \begin{equation*}
  %     \begin{tikzcd}
  %       \Rel{\Cat{E}}
  %       \dar[swap]{\Rel{P}}
  %       \rar{}
  %       \arrow[phantom]{dr}[very near start]{\lrcorner}
  %       & \Cat{E} \dar{P} \\
  %       \Cat{B} \rar{\Delta}
  %       & \Cat{B}
  %     \end{tikzcd}
  %   \end{equation*}
  % \end{wrapfigure}
    Then we can reinterpret $F$ and $G_u$ by
    \begin{align*}
      & F, G_u : \Rel{\TCat}_{\Str{A}} \to
      \Rel{\TCat}_{\Str{A}} \times \Rel{\TCat}_{\Str{A}} \\
      & F = \prodArr{\reidx[\#]{\head} \circ \, K_{\Eq(A)}, \reidx[\#]{\tail}}
      \qquad \text{and} \qquad
      G_u = \prodArr{\reidx[\#]{\id_{\Str{A}}}, \reidx[\#]{\id_{\Str{A}}}},
    \end{align*}
    where $\reidx[\#]{(-)}$ is reindexing in $\Rel{\TCat}$.
    The final $(G_u, F)$-dialgebra is a pair of morphisms
    \begin{equation*}
      (\head^{\sim}_A : \mathrm{Bisim}_A \to \reidx[\#]{\head}(\Eq(A)),
      \tail^{\sim}_A : \mathrm{Bisim}_A \to \reidx[\#]{\tail}(\mathrm{Bisim}_A)).
    \end{equation*}
    $\mathrm{Bisim}_A$ should be thought of to consist of all bisimilarity
    proofs.
    Coinductive extensions yield the usual coinduction proof principle,
    allowing us to prove bisimilarity by establishing a bisimulation relation
    $R \in \Cat{\Rel{E}}_{\Str{A}}$ together with $h : R \to \reidx[\#]{\head}(\Eq(A))$
    and $t : R \to \reidx[\#]{\tail}(R)$, saying that the heads of related
    streams are equal and that the tails of related streams are again related.

    % Later, we will see how bisimilarity arises from lifting the stream
    % signature to relations.
    % \todo{Where?}
  \end{enumerate}
\end{example}

The last example, we give, shall illustrate the additional
capabilities of CDTs in the present setup over those currently
available in Agda.
However, one should note that coinductive types in Agda
provide extra power in the sense that destructors can refer to each other.
This is equivalent to having a strong
coproduct~\cite[Sec. 10.1 and Def. 10.5.2]{Jacobs1999-CLTT}, which we do not
require in the setup of this work and thus
A proof of this equivalence is left out because of space constraints.
\begin{example}
  \label{ex:pstr}
  A partial stream is a stream together with a, possibly infinite, depth up to
  which it is defined.
  Assume that there is an object $\EN$ of natural numbers extended with infinity
  and a successor map $s_\infty : \EN \to \EN$ in $\BCat$, we will see
  how these can be defined below.
  Then partial streams correspond to the following type declaration.
  \begin{lstlisting}[language=Agda,mathescape=true,columns=flexible]
codata PStr (A : Set) : $\EN$ $\to$ Set where
  hd : (n : $\EN$) $\to$ PStr ($s_\infty$ n) $\to$ A
  tl : (n : $\EN$) $\to$ PStr ($s_\infty$ n) $\to$ PStr n
  \end{lstlisting}
  In an explicit, set-theoretic notation, we can define them as a family
  indexed by $n \in \EN$:
  \begin{equation*}
    \PStr{A}_n =
    \setDef{s : \parr{\N}{A}}{
      \forall k < n. \, k \in \dom s
      \wedge \forall k \geq n. \, k \not\in \dom s},
  \end{equation*}
  where the order on $\EN$ is given by extending that of the natural numbers
  with $\infty$ as strict top element, i.e., such that $k < \infty$ for all
  $k \in \N$.

  The interpretation of $\PStr{A}$ for $A \in \Cat{P}_{\T}$ in a data type
  complete fibration is given, similarly to vectors, as the carrier of the final
  $(G_u, F)$-dialgebra, where
  \begin{align*}
    G_u, F : \Cat{P}_{\EN} \to \Cat{P}_{\EN} \times \Cat{P}_{\EN}
    \qquad G_u = \prodArr*{\cramped{\reidx{s_\infty}, \reidx{s_\infty}}}
    \qquad F = \, \cramped{\prodArr*{K_{\overline{A}}^{\EN}, \Id}}
  \end{align*}
  and $\overline{A} = \reidx{!_{\EN}}(A) \in \Cat{P}_{\EN}$ is the weakening of
  $A$ using $!_{\EN} : \EN \to \T$.
  The idea of this signature is that the head and tail of partial streams are
  defined only on those partial streams that are defined in, at least, the first
  position.
  On set families, partial streams are given by the dialgebra
  $\outNu = (\head, \tail)$ with
  $\head_n : \PStr{A}_{(s_\infty \, n)} \to A$ and
  $\tail_n : \PStr{A}_{(s_\infty \, n)} \to \PStr{A}_n$
  for every $n \in \EN$.

  We can make this construction functorial in $A$, using the same ``trick''
  as for sums and products.
  To this end, we define the functor
  $H : \Cat{P}_{\T} \times \Cat{P}_{\EN} \to \Cat{P}_{\EN} \times \Cat{P}_{\EN}$
  with $H = \prodArr*{\cramped{!_{\EN} \circ \pi_1, \pi_2}}$, where
  $\pi_1$ and $\pi_2$ are corresponding projection functors,
  so that $H(A, X) = F(X)$.
  This gives, by data type completeness, rise to a functor
  $\nu (\funCur{G_u}, \funCur{F}) : \Cat{P}_{\EN} \to \Cat{P}_{\EN}$, which we
  denote by $\PStr{}$, together with a pair $(\head, \tail)$ of natural
  transformations.
  \qed
\end{example}

We have seen in the examples above that we would often like to use a data type
again as index, which means that we need a mechanism to turn a data type in
$\TCat$ into an index in $\BCat$.
This is provided by, so called, \emph{comprehension}.
\begin{definition}[See {\cite[Lem. 1.8.8, Def. 10.4.7]{Jacobs1999-CLTT}} and
  \cite{Ghani-IndexedInduct}]
  \label{def:CCU}
  Let $P : \TCat \to \BCat$ be a fibration.
  If each fibre $\Cat{P}_I$ has a final object $\T_I$ and these are preserved by
  reindexing, then there is a fibred \emph{final object functor}
  $\K{} : \BCat \to \TCat$. (Note that then $P(\K{I}) = I$.)
  $P$ is a \emph{comprehension category with unit} (CCU), if $\K{}$ has a
  right adjoint $\compr{-} : \TCat \to \BCat$, the \emph{comprehension}.
  This gives rise to a functor $\mathcal{P} : \TCat \to \ArrC{\BCat}$ into the
  arrow category over $\BCat$, by mapping
  $A \mapsto P(\varepsilon_A) : \compr{A} \to P(A)$, where
  $\nat{\varepsilon}{\K{\compr{-}}}{\Id}$ is the counit of
  $\K{} \dashv \compr{-}$.
  We often denote $\mathcal{P}(A)$ by $\pi_A$ and call it the \emph{projection}
  of $A$.
  Finally, $P$ is said to be a \emph{full} CCU, if $\mathcal{P}$ is full.
\end{definition}

Note that, in a data type complete category, we can define final objects in each
fibre, the preservation of them needs to be required separately.

\begin{example}
  In $\Fam{\SetC}$, the final object functor is given by
  $\K{I} = (I, \{\T\}_{i \in I})$, where $\T$ is the singleton set.
  Comprehension is defined to be
  $\compr{(I, X)} = \coprod_{i \in I} X_i$ and the projections $\pi_I$ map
  then an element of $\coprod_{i \in I} X_i$ to its component $i \in I$.
\end{example}

Using comprehension, we can give a general account to dependent data types.
\begin{definition}
  \label{def:DTCC}
  We say that a fibration $P : \TCat \to \BCat$ is a
  \emph{data type closed category} (DTCC), if it is a CCU, has a terminal object
  in $\BCat$ and is data type complete.
\end{definition}

As already mentioned, the purpose of introducing comprehension is that it allows
us to use data types defined in $\TCat$ again as index.
The terminal object in $\BCat$ is used to introduce data types without
dependencies, like the natural numbers.
Let us reiterate on \iExRef{pstr}.
\begin{example}
  Recall that we assumed the existence of extended naturals $\EN$ and the
  successor map $s_\infty$ on them to define partial streams.
  We are now in the position to define, in a data type closed category,
  everything from scratch as follows.

  Having defined $+ : \Cat{P}_\T \times \Cat{P}_\T \to \Cat{P}_\T$, see
  \iThmRef{adjunctions-from-dt}, we put $\EN = \nu (\Id, \T + \Id)$ and find
  the predecessor $\pred$ as the final dialgebra on $\EN$.
  The successor $s_\infty$ arises as the coinductive extension
  $(\EN, \kappa_2) \to (\EN, \pred)$, where $\kappa_2$ is the coproduct
  inclusion.
  Partial streams
  $\PStr{} : \Cat{P}_{\compr{\EN}} \to \Cat{P}_{\compr{\EN}}$ are then given, as in \iExRef{pstr},
  by the final $(\funCur{G}, \funCur{F})$-dialgebra with
  $G = \prodArr*{\cramped{\reidx{\compr{s_\infty}}, \reidx{\compr{s_\infty}}}}$
  and $F = \prodArr*{\cramped{!_{\EN} \circ \pi_1, \pi_2}}$.
  \qed
\end{example}

\section{Constructing Data Types}
\label{sec:construct-dt}

In this section, we show how some data types can be constructed
through polynomial functors, where I draw from the vast amount of work on
polynomial functors that exists in the literature,
see~\cite{AbbottContainer,GambinoKockPolyFunctors}.
The construction works by, first, reducing dialgebras to (co)algebras
and, second, constructing the necessary initial algebras and final coalgebras
as fixed points of polynomial functors analogously to the construction
of strictly positive types in \cite{AbbottContainer}.
This result works thus far only for data types that, if at all, only use
dependent coinductive types at the top-level.
Nesting of dependent inductive and non-dependent coinductive types works,
however, in full generality.

Before we come to polynomial functors and their fixed points, we show that
inductive and coinductive data types actually correspond to initial algebras and final
coalgebras, respectively.
\begin{theorem}
  \label{thm:iso-dialg-alg}
  Let $P : \Cat{E} \to \Cat{B}$ be a fibration with fibrewise
  coproducts and dependent sums.
  If $(F, u)$ with $F : \Cat{P}_I \to \Cat{P}_{J_1} \times \dotsm \times \Cat{P}_{J_n}$
  is a signature,
  then there is an isomorphism
  \begin{equation*}
    \DialgC{F,G_u} \cong
    \AlgC{\coprod_{u_1} \circ F_1 +_I \dotsm +_I \coprod_{u_n} \circ F_n}
  \end{equation*}
  where $F_k = \pi_k \circ F$ is the $k$th component of $F$.
  In particular, existence of inductive data types and initial algebras
  coincide.
  Dually, if $P$ has fibrewise and dependent products, then
  \begin{equation*}
    \DialgC{G_u,F} \cong
    \CoalgC{\prod_{u_1} \circ F_1 \times_I \dotsm \times_I \prod_{u_n} \circ F_n}.
  \end{equation*}
  In particular, existence of coinductive data types and final coalgebras
  coincide.
\end{theorem}
\begin{proof}
  The first result is given by a simple application of the adjunctions
  $\coprod_{k=1}^n \dashv \Delta_n$ between the (fibrewise) coproduct and the
  diagonal, and $\coprod_{u_k} \dashv \reidx{u_k}$:
  \begin{equation*}
    \begin{adjunction}
%    \SwapAboveDisplaySkip
      & FX & G_uX
      & \text{(in } \Cat{P}_{J_1} \times \dotsm \times \Cat{P}_{J_n} \text{)} \\
      \hhline{====}
      & (\coprod_{u_1} (F_1 X), \dotsc, \coprod_{u_n}(F_n X)) & \Delta_n X
      & \text{(in } \Cat{P}_I^n \text{)} \\
      \hhline{====}
      & \coprod_{k=1}^n \coprod_{u_k} (F_k X) & X
      & \text{(in } \Cat{P}_I \text{)}
    \end{adjunction}
  \end{equation*}
  % \begin{equation*}
  %   \alwaysDoubleLine
  %   \AxiomC{$F X \to G X \qquad$
  %     (in $\Cat{P}_{J_1} \times \dotsm \times \Cat{P}_{J_n}$)}
  %   \UnaryInfC{$(\coprod_{u_1} (F_1 X), \dotsc, \coprod_{u_n}(F_n X)) \to
  %     \Delta_n X \qquad $ (in $\Cat{P}_I^n$)}
  %   \UnaryInfC{$\coprod_{k=1}^n \coprod_{u_k} (F_k X) \to X \qquad$
  %     (in $\Cat{P}_I$)}
  %   \DisplayProof
  % \end{equation*}
  That (di)algebra homomorphisms are preserved follows at once from
  naturality of the used Hom-set isomorphisms.
  % $\Hom[\Cat{P}_A]{\coprod_{k=1}^n X_k}{X} \cong
  % \Hom[\Cat{P}_A^n]{(X_1, \dotsc, X_k)}{\Delta_n X}$
  % and
  % $\Hom[\Cat{P}_A]{\coprod_{u_k} X}{Y} \cong \Hom[\Cat{P}_{B_k}]{X}{\reidx{u_k} Y}$.
  The correspondence for coinductive types follows by duality.
  \qedhere
\end{proof}

To be able to reuse existing work, we work in the following with the codomain
fibration $\cod : \ArrC{\BCat} \to \BCat$ for a category $\BCat$ with pullbacks.
Moreover, we assume that $\BCat$ is locally Cartesian closed,
which is equivalent to say that  $\cod : \ArrC{\BCat} \to \BCat$ is a closed
comprehension category, that is, it is a full CCU with products and coproducts,
and $\BCat$ has a final object, see \cite[Thm 10.5.5]{Jacobs1999-CLTT}.
Finally, we need disjoint coproducts in $\BCat$, which gives us an equivalence
$\slice{\BCat}{I + J} \simeq \slice{\BCat}{I} \times \slice{\BCat}{J}$,
see \cite[Prop. 1.5.4]{Jacobs1999-CLTT}.

\begin{definition}
  A \emph{dependent polynomial} $P$ indexed by $I$ on variables indexed by $J$
  is given by a triple of morphisms
  \begin{equation*}
    \begin{tikzcd}[row sep=0.1em]
      & B \arrow{dl}[swap]{s} \rar{f} & A \arrow{dr}{t} & \\
      J & & & I
    \end{tikzcd}
  \end{equation*}
  If $J = I = \T$, $f$ is said to be a
  \emph{(non-dependent) polynomial}.
  The extension of $P$ is given by the composite
  \begin{equation*}
    \polySem{P} =
    \slice{\BCat}{J} \xrightarrow{\reidx{s}}
    \slice{\BCat}{B} \xrightarrow{\prod_f}
    \slice{\BCat}{A} \xrightarrow{\coprod_t}
    \slice{\BCat}{I},
  \end{equation*}
  which we denote by $\polySem{f}$ if $f$ is non-dependent.
  A functor $F : \slice{\BCat}{J} \to \slice{\BCat}{I}$ is a
  \emph{dependent polynomial functor}, if there is a dependent polynomial $P$
  such that $F \cong \polySem{P}$.
\end{definition}

\begin{remark}
  Note that polynomials are called \emph{containers} by
  Abbott et al.~\cite{AbbottContainer,AbbottThesis}, and
  a polynomial $P = \poly[1]{!}[B]{f}[A]{!}[1]$ would be written as
  $A \triangleright f$.
  % The corresponding extension is given then by
  % \begin{equation*}
  %   X \mapsto \coprod_{a : A} X^{B_a},
  % \end{equation*}
  % where $B_a$ is the fibre of $f$ above $a$.
  Container morphisms, however, are different from those of dependent
  polynomials, as the latter correspond strong natural
  transformations~\cite[Prop. 2.9]{GambinoKockPolyFunctors},
  whereas the former are in exact correspondence with all
  natural transformations between extensions~\cite[Thm. 3.4]{AbbottContainer}.
\end{remark}

Because of this relation, we will apply results for containers that do not
involve morphisms to polynomials.
In particular, \cite[Prop. 4.1]{AbbottContainer} gives us that we can construct
final coalgebras for polynomial functors from initial algebras for polynomial
functors.
The former are called \emph{M-types} and are denoted by $M_f$ for $f : A \to B$,
whereas the latter are \emph{W-types} and denoted by $W_f$.

\begin{assumption}
  We assume that $\BCat$ is closed under the formation of W-types, thus is
  a \emph{Martin-Löf category} in the terminology of \cite{AbbottContainer}.
\end{assumption}
By the above remark, $\BCat$ then also has all M-types.

Analogously to how \cite[Thm. 12]{Gambino-WTypesPolyFunc} extends
\cite[Prop. 3.8]{Moerdijk-WellFoundedTrees},
we extend here \cite[Thm 3.3]{vdBerg-Non-wellfoundedTrees}.
As it was pointed out by one reviewer, this result is actually
in~\cite{vandenberg_nonwellfounded_2007}, the published version of
\cite{vdBerg-Non-wellfoundedTrees}.
\begin{theorem}
  \label{thm:dep-final-coalg-from-non-dep}
  If $\BCat$ has finite limits, then every dependent polynomial functor has a
  final coalgebra in $\slice{\BCat}{I}$.
\end{theorem}
\begin{proof}
  Let $P = \poly[I]{s}[B]{f}[A]{t}[I]$ be a dependent polynomial, we construct,
  analogously to \cite{Gambino-WTypesPolyFunc} the final coalgebra $V$ of
  $\polySem{P}$ as an equaliser as in the following diagram, in which
  $f \times I$ is a shorthand for
  $B \times I \xrightarrow{f \times \id_I} A \times I$ and
  $M_{f \times I}$ is the carrier of the final
  $\polySem{f \times I}$-coalgebra.
  \begin{equation*}
    \begin{tikzcd}[column sep=1.5cm]
      V \rar{g}
      & M_f \arrow[start anchor=north east, end anchor=north west]{r}{u_1}
      \arrow[start anchor=south east, end anchor=south west]{r}[swap]{u_2}
      & M_{f\times I}
    \end{tikzcd}
  \end{equation*}

  First, we give $u_1$ and $u_2$, whose definitions are summarised in the
  following diagrams.
  \begin{equation*}
    \begin{tikzcd}[column sep=3.5em]
      M_f \rar{u_1} \dar{\xi_f}
      & M_{f \times I} \arrow{dd}{\xi_{f \times I}} \\
      \polySem{f}(M_f) \dar{p_{M_f}}
      & \\
      \polySem{f \times I}(M_f) \rar{\polySem{f \times I}(u_1)}
      & \polySem{f \times I}(M_{f \times I})
    \end{tikzcd}
    \begin{tikzcd}[column sep=1em]
      M_f \rar{u_1} \arrow[bend left=20]{rr}{u_2}
      & M_{f \times I} \dar{\xi_{f \times I}} \rar{\psi}
      &[2.5em] M_{f \times I} \arrow{dd}{\xi_{f \times I}} \\
      & \polySem{f \times I}(M_{f \times I}) \dar{\Sigma_{A \times I}K} & \\
      & \polySem{f \times I}(M_{f \times I} \times B)
      \rar{\polySem{f \times I}(\phi)}
      & \polySem{f \times I}(M_{f \times I})
    \end{tikzcd}
  \end{equation*}
  These diagrams shall indicate that $u_1$ is given as coinductive extensions
  and $\psi$ as one-step definition (which can be defined using coproducts)%see \appRef{proofs-construct-dt})
  , using that $M_{f \times I}$ is a final coalgebra.
  The maps involved in the diagram are given as follows, which we sometimes
  spell out in the internal language of $\cod$, see for
  example~\cite{AbbottThesis}, as this is sometimes more readable.
  \begin{itemize}
  \item $p : \Sigma_A \Pi_{f} \Rightarrow \Sigma_{A \times I} \Pi_{f \times I}$
    is the natural transformation that maps
    $(a, v)$ to $(a, t(a), v)$.
    It is given by the extension
    $\polySem{\alpha, \beta} : \polySem{f} \Rightarrow \polySem{f \times I}$
    of the morphism of polynomials~\cite{GambinoKockPolyFunctors}
    \begin{equation*}
      \begin{tikzcd}
        B \rar{f} \dar[swap]{\beta} \pullback & A \dar{\alpha} \\
        B \times I \rar{f \times I} & A \times I
      \end{tikzcd}
    \end{equation*}
    where $\alpha = \prodArr{\id, t}$ and $\beta = \prodArr{\id, t \circ f}$.
  \item The map
    $K : \Pi_{f \times I}(M_{f \times I})
    \to \Pi_{f \times I}(M_{f \times I} \times B)$ is given
    as transpose of
    $\prodArr{\varepsilon_{M_{f \times I}}, \pi_1 \circ \pi}
    : \reidx{(f \times I)}(\Pi_{f \times I}(M_{f \times I}))
    \to M_{f \times I} \times B$,
    where $\varepsilon$ is the counit of the product (evaluation) and
    $\pi$ is the context projection.
    In the internal language $K$ is given by
    $K \, v = \lambda (b,i). (v \, (b,i), b)$.
  \item $\phi : M_{f \times I} \times B \to M_{f \times I}$ is constructed as
    coinductive extension as in the following diagram
  \begin{equation*}
    \begin{tikzcd}[column sep=3.5em]
      M_{f \times I} \times B
      \dar{\xi_{f \times I} \times \id}
      \rar{\phi}
      & M_{f \times I} \arrow{dd}{\xi_{f \times I}} \\
      \polySem{f \times I}(M_{f \times I}) \times B \dar{e}
      & \\
      \polySem{f \times I}(M_{f \times I} \times B)
      \rar{\polySem{f \times I}(\phi)}
      & \polySem{f \times I}(M_{f \times I})
    \end{tikzcd}
  \end{equation*}
  Here $e$ is given by
  $e((a, i, v),b) = (a, s \, b, \lambda (b',s \, b). (v \, (b',i), b'))$.
  \end{itemize}

  The important property, which allows us to prove that
  $\xi_f : M_f \to \polySem{f}(M_f)$ restricts to
  $\xi' : V \to \polySem{P}(V)$ and that $\xi'$ is a final coalgebra, is
  that $x : V_i$ $\iff$
  $ \xi_f \, x = (a : A , v : \Pi_f M_f), \, t \, a = i \text{ and }
  (\forall b : B. f \, b = a \Rightarrow v \, b : V_{s \, b})$.
  The direction from left to right is given by simple a calculation,
  whereas the other direction can be proved by establishing a bisimulation and
  between $u_1 \, x$ and $u_2 \, x$.

  Hence $V$, given as a subobject of $M_f$, is
  indeed the final $\polySem{P}$-coalgebra in $\slice{\BCat}{I}$.
\end{proof}

Combining this with \cite[Prop. 4.1]{AbbottContainer}, we have that the
existence of final coalgebras for dependent polynomial functors follows
from the existence of initial algebras of (non-dependent) polynomial
functors.
This gives us the possibility of interpreting non-nested fixed points in
any Martin-Löf category as follows.

First, we observe that the equivalence
$\slice{\BCat}{I + J} \simeq \slice{\BCat}{I} \times \slice{\BCat}{J}$ allows us
to rewrite the functors from \iThmRef{iso-dialg-alg} to
a form that is closer to polynomial functors:
\begin{align*}
  \coprod_{u_1} \circ \, F_1 +_I \dotsm +_I \coprod_{u_n} \circ F_n
  & \cong \coprod_u F' \\
  \prod_{u_1} \circ \, F_1 \times_I \dotsm \times_I \prod_{u_n} \circ F_n
  & \cong \prod_u F',
\end{align*}
where
$J = J_1 + \dotsm + J_n$,
$u : J \to I$ is given by the cotupling
$\coprodArr{u_1, \dotsc, u_n}$ and $F' : \slice{\BCat}{I} \to \slice{\BCat}{J}$
is given by
$F' = \prodArr{F_1, \dotsc, F_n} : \slice{\BCat}{I}
\to \prod_{i =1}^n \slice{\BCat}{J_i} \simeq \slice{\BCat}{J}$.
Thus, if we establish that $F'$ is a polynomial functor, we get that
$\coprod_u F'$ and $\prod_u F'$ are polynomial functors,
see \cite{AbbottThesis}.
For non-nested fixed points, that is, $F_k$ is either
a constant functor, given by composition or reindexing, this is immediate,
as dependent polynomials can be composed and are closed
under constant functors and reindexing, see~\cite{GambinoKockPolyFunctors}.

We say that a dependent polynomial is \emph{parametric}, if it is of the
following form.
\begin{equation*}
  \begin{tikzcd}
    K + I & B \lar[swap]{s} \rar{f} & A \rar{t} & I
  \end{tikzcd}
\end{equation*}
Such polynomials represent polynomial functors
$\slice{\Cat{B}}{K} \times \slice{\Cat{B}}{I} \to \slice{\Cat{B}}{I}$ and allow
us speak about nested fixed points just as we have done in~\iSecRef{fib-dialg}.
What thus remains is that fixed points of parametric dependent polynomial
functors, in the sense of \iSecRef{fib-dialg}, are again dependent polynomial
functors.

The proof of this is literally the same as that for
containers~\cite[Sec.~5.3-5.5]{AbbottThesis}
or non-dependent polynomials~\cite{Gambino-WTypesPolyFunc}, except that we need
to check some extra conditions regarding the indexing.
\begin{theorem}
  Initial algebras and final coalgebras of parametric, dependent polynomial
  functors are again dependent polynomial functors.
\end{theorem}
\begin{proof}
  Let
  \begin{equation*}
    \begin{tikzcd}[row sep=1pt]
      F = J & B \lar[swap]{s} \rar{f} & A \rar{t} & I \\
      G = I & D \lar[swap]{u} \rar{g} & C \rar{v} & I
    \end{tikzcd}
  \end{equation*}
  be dependent polynomials and $H(X, Y) = \polySem{F} \times_I \polySem{G}$ be
  the parametric dependent polynomial functor in question.
  Assuming that there is a polynomial
  \begin{equation*}
    \begin{tikzcd}
      J & Q \lar[swap]{x} \rar{h} & P \rar{y} & I
    \end{tikzcd}
  \end{equation*}
  so that for $K = \coprod_y \prod_h \reidx{x}$ we have
  $K(X) \cong H(X, K(X))$, we can calculate, as in~\cite{AbbottThesis},
  that we need to have isomorphisms
  \begin{align*}
    \psi : A \times_I \polySem{G}(P) \cong P \\
    \varphi : B + \coprod_g \reidx{\varepsilon}Q \cong \reidx{\psi}(Q)
  \end{align*}
  where $B + \coprod_g \reidx{\varepsilon}Q$ is, as in loc. cit., is an
  abbreviation for $B_a + \coprod_{d : D_c} Q (r \, d)$ in the context
  $(a, (c, r)) : A \times_I \polySem{G}(P)$.
  If $K(X)$ shall be an initial algebra, $\psi$ must an initial algebra as well,
  whereas if $K(X)$ shall be a final coalgebra, $\psi$ must be one.
  The isomorphism $\varphi$ is given as the initial
  $\reidx{(\psi^{-1})} (B + \coprod_g \reidx{\varepsilon})$-algebra in both
  cases, see~\cite{AbbottThesis}.
  This we use to define $x : Q \to J$ as the inductive extension of the map
  $\coprodArr{s, \pi_2} :
  \reidx{(\psi^{-1})}(B + \coprod_g \reidx{\varepsilon}J) \to J$.
  Given these definitions, the following diagrams commute.
  \begin{equation*}
    \begin{tikzcd}
      A \times_I \polySem{G}(P) \arrow{rr}{\psi} \arrow{dr}{}
      & & P \arrow{dl}{y} \\
      & I &
    \end{tikzcd}
    \qquad
    \begin{tikzcd}
      B + \coprod_g \reidx{\varepsilon}Q
      \arrow{rr}{\varphi} \arrow{dr}[swap]{[s, x \circ \pi_2]}
      & & \reidx{\psi} Q \arrow{dl}{} \\
      & J &
    \end{tikzcd}
  \end{equation*}
  This gives us that the isomorphism given in the proofs of
  \cite[Prop. 5.3.1, Prop. 5.4.2]{AbbottThesis} also work for the dependent
  polynomial case.
  The rest of the proofs in loc. cit. go then through, as well.
  Thus $K$ is in both cases again given by a dependent polynomial.
\end{proof}

% For initial algebras, this is an instance
% of~\cite[Thm. 4.5]{GambinoKockPolyFunctors}, where Gambino and Kock have proved
% that the free monad for a polynomial is again given by a polynomial.
% Unfortunately, the proof given in loc. cit. cannot be adapted easily, as the
% reindexing map of the dependent polynomial constructed for the fixed point is
% given by recursion, which we cannot do in the coinductive case.
% It remains open for now whether final coalgebras for parametric polynomials
% are polynomial functors themselves.
% It should be noted however that container are closed under taking final
% coalgebras, see~\cite[Prop. 5.4]{AbbottContainer}.

Summing up, we are left with the following result.
\begin{corollary}
  All data types for strictly positive signatures can be constructed in any
  Martin-Löf category.
  % Data types that only have dependent coinductive types at the top level,
  % that is, are of the form $\nu (\funCur{G_u}, \funCur{F})$ for some
  % $F : \slice{B}{K} \times \slice{B}{I} \to \slice{B}{J}$ as is given
  % in \iDefRef{dt-complete} but only using non-dependent coinductive
  % data types, can be constructed in any Martin-Löf category.
\end{corollary}

Let us see, by means of an example, how the construction in the proof of
\iThmRef{dep-final-coalg-from-non-dep} works intuitively.
\begin{example}
  Recall from \iExRef{pstr} that partial streams are given by the declaration
  \begin{lstlisting}[language=Agda,mathescape=true,columns=flexible]
codata PStr (A : Set) : $\EN$ $\to$ Set where
  hd : (n : $\EN$) $\to$ PStr ($s_\infty$ n) $\to$ A
  tl : (n : $\EN$) $\to$ PStr ($s_\infty$ n) $\to$ PStr n
  \end{lstlisting}
  By \iThmRef{iso-dialg-alg}, we can construct $\PStr{}$ as the final
  coalgebra of
  $F : \slice{\BCat}{\T} \times \slice{\BCat}{\EN} \to \slice{\BCat}{\EN}$
  with $F(A, X) = \prod_{s_\infty} \reidx{!} A \times \prod_{s_\infty} X$.
  Note that $F$ is isomorphic to
  $\slice{\BCat}{\T} \times \slice{\BCat}{\EN}
  \simeq \slice{\BCat}{\T + \EN} \xrightarrow{\polySem{P}} \slice{\BCat}{\EN}$,
  where $P$ is the polynomial
  \begin{equation*}
    P = \T + \EN
    \xleftarrow{g} 2 \times \EN
    \xrightarrow{f} \EN
    \xrightarrow{\id} \EN
    \qquad g(i, k) =
    \begin{cases}
      \kappa_1 \ast, & i = 1 \\
      \kappa_2 k, & i = 2
    \end{cases}
    \qquad f(i, k) = s_\infty \, k.
  \end{equation*}
  If we now fix an object $A \in \slice{\BCat}{\T}$, then
  $F(A, -) \cong \polySem{P'}$ for the polynomial $P'$ given by
  \begin{equation*}
    P' = \EN
    \xleftarrow{\pi} \sum_\EN \sum_{s_\infty} \prod_{s_\infty} \reidx{!} A
    \xrightarrow{f'} \sum_\EN \prod_{s_\infty} \reidx{!} A
    \xrightarrow{\pi} \EN,
  \end{equation*}
  where $\pi$ is the projection on the index of a dependent sum
  and $f'(n, (s_\infty \, n, v)) = (s_\infty \, n, v)$.

  Recall that we construct in \iThmRef{dep-final-coalg-from-non-dep} the final
  coalgebra of $\polySem{P'}$ as a subobject of $M_{f'}$.
  Below, we present three trees that are elements of $M_{f'}$, where only
  the second and third are actually selected by the equaliser taken in
  \iThmRef{dep-final-coalg-from-non-dep}.
  \forestset{
    tree edge label/.style 2 args={
      edge label={node[midway, font=\scriptsize, #1]{#2}}, %
    },
    ntree edge label/.style n args={3}{
      edge label={node (#2) [midway, font=\scriptsize, #1] {#3}},
    },
  }
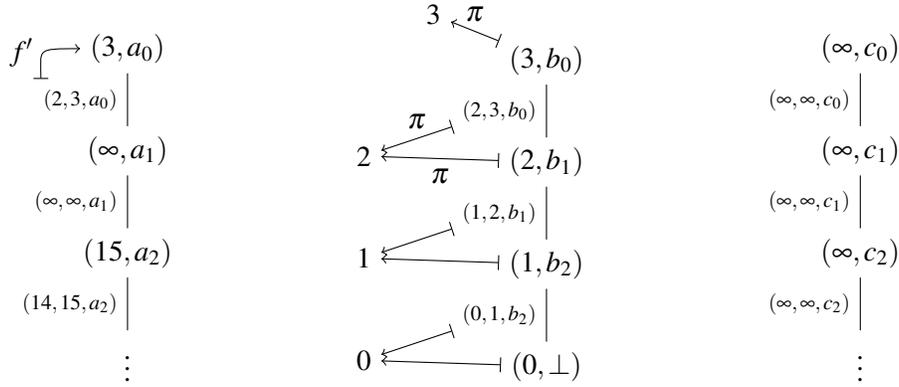
  \begin{center}
    \begin{forest}
      for tree={l sep=20pt},
      [{$(3, a_0)$}, name={t1-1}
      [{$(\infty, a_1)$}, ntree edge label={left}{t1-1c}{$(2, 3, a_0)$}
      [{$(15, a_2)$}, tree edge label={left}{$(\infty, \infty, a_1)$}
      [{$\vdots$}, tree edge label={left}{$(14, 15, a_2)$}
      ] ] ] ]
      \draw[|->] ($(t1-1c.north west) + (0.1,0)$)
      .. controls +(-90:-0.4) and +(0:-0.6)
      .. node[left] {$f'$} (t1-1.west);
    \end{forest}
    \hspace{2cm}
    \begin{forest}
      for tree={l sep=20pt}
      [{$(3, b_0)$}, name={t2-1}
      [{$(2, b_1)$}, name={t2-2}, ntree edge label={left}{t2-1br}{$(2, 3, b_0)$}
      [{$(1, b_2)$}, name={t2-3}, ntree edge label={left}{t2-2br}{$(1, 2, b_1)$}
      [{$(0, \bot)$}, name={t2-4}, ntree edge label={left}{t2-3br}{$(0, 1, b_2)$}
      ] ] ] ]
      \node (t2-idx) at ($(t2-1) + (-1.5, 0.6)$) {$3$};
      \draw[|->] (t2-1) -- node[above] {$\pi$} (t2-idx);
      \node (t2-2-idx) at ($(t2-1br) + (-1.8,-0.6)$) {$2$};
      \draw[|->] (t2-1br) -- node[above] {$\pi$} (t2-2-idx);
      \draw[|->] (t2-2) -- node[below] {$\pi$} (t2-2-idx);
      \node (t2-3-idx) at ($(t2-2br) + (-1.8,-0.6)$) {$1$};
      \draw[|->] (t2-2br) -- (t2-3-idx);
      \draw[|->] (t2-3) -- (t2-3-idx);
      \node (t2-4-idx) at ($(t2-3br) + (-1.8,-0.6)$) {$0$};
      \draw[|->] (t2-3br) -- (t2-4-idx);
      \draw[|->] (t2-4) -- (t2-4-idx);
    \end{forest}
    \hspace{2cm}
    \begin{forest}
      for tree={l sep=20pt}
      [{$(\infty, c_0)$}
      [{$(\infty, c_1)$}, tree edge label={left}{$(\infty, \infty, c_0)$}
      [{$(\infty, c_2)$}, tree edge label={left}{$(\infty, \infty, c_1)$}
      [{$\vdots$}, tree edge label={left}{$(\infty, \infty, c_2)$}
      ] ] ] ]
    \end{forest}
  \end{center}
  Here we denote a pair $(k, v) : \sum_\EN \prod_{s_\infty} \reidx{!} A$ with
  $k = s_\infty \, n$ and $v \, n = a$ by $(k, a)$,
  or if $k = 0$ by $(0, \bot)$.
  Moreover, we indicate the matching of indices in the second tree, which is
  used to form the equaliser.
  Note that the second tree is an element of $\PStr{A} \, 3$, whereas the
  third is in $\PStr{A} \, \infty$.
  \qed
\end{example}

\section{Conclusion and Future Work}
\label{sec:concl}

We have seen how dependent inductive and coinductive types with type
constructors, in the style of Agda, can be given semantics in terms of data
type closed categories (DTCC), with the restriction that destructors of
coinductive types are not allowed to refer to each other.
This situation is summed up in the following table.
\begin{center}
  \begin{tabular}[t]{l|p{0.65\textwidth}}
    Condition & Use/Implications \\ \hline
    Cloven fibration & Definition of signatures and data types \\
    Data type completeness
      & Construction of types indexed by objects in base (e.g., vectors for
      $\N \in \BCat$) and types agnostic of indices (e.g., initial and final
      objects, sums and products) \\
    Data type closedness
      & Constructed types as index; Full interpretation of data types % \\
    % Beck-Chevalley condition
    %   & Preservation of type structure  under reindexing \\
    % Full comprehension
    %   & Lifting of functors for induction/dependent recursion \\
    % Strong sum & Induction principle
  \end{tabular}
\end{center}
Moreover, we have shown that a large part of these data types can be constructed
as fixed points of polynomial functors.

Let us finish by discussing directions for future work.
% First, the question of whether all data types can be constructed through
% polynomials remains open, which is, however, likely to hold.
First, a full interpretation of syntactic data types has also still to be
carried out.
Here one has to be careful with type equality, which is usually dealt with
using split fibrations and a Beck-Chevalley condition.
The latter can be defined generally for the data types of this work, in needs
to be checked, however, whether this condition is sufficient for giving
a sound interpretation.
Finally, the idea of using dialgebras has found its way into the syntax
of higher inductive types~\cite{Capriotti-MutualHIT}, though in that work the
used format of dialgebras is likely to be too liberal to guarantee the
existence of semantics.
The reason is that the shape of dialgebras used in the present work ensures
that we can construct data types from (co)coalgebras, whereas this
is not the case in~\cite{Capriotti-MutualHIT}.
Thus it is to be investigated what the right notion of dialgebras is for
capturing higher (co)inductive types, such that their semantics in terms of
trees can always be constructed.

%%% Local Variables:
%%% TeX-master: "../FibDialg-FICS"
%%% ispell-local-dictionary: "british"
%%% mode: latex
%%% End:

\bibliographystyle{eptcs}
\bibliography{FibDialg}

% \clearpage
% \appendix
% \input{content/appendix}

\end{document}